\documentclass[aip,jmp,floatfix,10pt,twoside,notitlepage]{revtex4-1}

 \usepackage[colorlinks,linkcolor=black,anchorcolor=green,citecolor=black]{hyperref}
\usepackage{amsfonts,amssymb}
\usepackage{amsmath,amsthm}
\usepackage{graphicx}
\usepackage{bm}% bold math
\usepackage{enumerate}
\usepackage[utf8]{inputenc}
\usepackage[english]{babel}

\usepackage{float}
\usepackage{mathrsfs}
\usepackage{upgreek}
\usepackage{subcaption}

\numberwithin{equation}{section}

\newcounter{fakecnt}[section]

  \newtheorem{theorem}{Theorem}[fakecnt]
  \newtheorem{lemma}[theorem]{Lemma}
  \newtheorem{corollary}[theorem]{Corollary}
  \newtheorem{proposition}[theorem]{Proposition}

 \theoremstyle{remark} 
  \newtheorem{remark}[theorem]{Remark}

\graphicspath{{./Graphics/}}

   \allowdisplaybreaks
              
\begin{document}

\title{Geometric criteria for the absence of effective many-body interactions in nonadditive hard particle mixtures}

\author{Ren{\'e} Wittmann}
\email{rene.wittmann@hhu.de}
\affiliation{Institut f\"ur Theoretische Physik II: Weiche Materie, Heinrich-Heine-Universit\"at D\"usseldorf, 40225 D\"usseldorf, Germany}

\author{Sabine Jansen}
\email{jansen@math.lmu.de}
\affiliation{Mathematisches Institut,  Ludwig-Maximilians-Universit\"at M\"unchen, 
\mbox{80333 Munich, Germany}; Munich Center for Quantum Science and Technology (MCQST), Schellingstr.~4, 80799 M{\"u}nchen, Germany}

\author{Hartmut L{\"o}wen}
\email{hlowen@hhu.de}
\affiliation{Institut f\"ur Theoretische Physik II: Weiche Materie, Heinrich-Heine-Universit\"at D\"usseldorf, 40225 D\"usseldorf, Germany}

\date{\today}   

\begin{abstract}
We consider a mixture of small and big classical particles in arbitrary spatial dimensions interacting via hard-body potentials with non-additive excluded-volume interactions. 
In particular, we focus on variants of the Asakura--Oosawa (AO) model where
the interaction between the small particles  is neglected
but the big-small and big-big interactions are present and can be condensed into an effective depletion interaction among the big particles alone.
The original AO model involves hard spherical particles in three spatial dimensions with
interaction diameters $\sigma_\text{pp}=0$, $\sigma_\text{cc}>0$ and $\sigma_\text{pc}>\sigma_\text{cc}/2$ respectively, 
where  $\sigma_{ij}$ with $\{i,j\}=\{\text{p},\text{c}\}$ (indicating the physical interpretation of the small and big particles as polymers (p) and colloids (c), respectively) is the minimum possible center-to-center distance between particle $i$ and particle $j$ allowed by the excluded-volume constraints.
It is common knowledge that there are only pairwise effective depletion interactions between the big particles if the geometric condition $\sigma_\text{pc}/\sigma_\text{cc} < 1/\sqrt{3}$ is fulfilled.
In this case, triplet and higher-order many body interactions are vanishing and the equilibrium statistics
of the binary mixture can exactly be mapped onto that of an effective one-component system with
the effective depletion pair-potential.
Here we prove this geometric criterion rigorously and generalize it to polydisperse mixtures and to anisotropic particle shapes in any dimension, providing
geometric criteria sufficient to guarantee the absence of triplet and higher-order many body interactions. 
For an external  hard wall confining the full mixture, we also give criteria which guarantee that the system can be mapped onto one with effective external one-body interactions.
\end{abstract}

\maketitle

\section{Introduction}

Mixtures of classical particles establish not only important model systems of statistical mechanics
\cite{Ballone1986,Dijkstra1998,dijkstra1999direct,Roth2001,Schmidt0,Schmidt2,Schmidt_2004,Schmidt2010,Ramon2014,Oettel2018,Smallenburg2020,Egorov,Ditz2021,Egorov2021}
but are also used to a large extent to describe colloidal and colloid--polymer
mixtures \cite{PuseyLes_Houches_1991,Tuinier_Lekkerkerker_textbook,Poon1,Poon2,Opdam2021,Denton2021,Gussmann2021,Tom2021,Gimperlein2021}.
One of the most famous non-additive models is that originally proposed by Asakura and Oosawa \cite{AO1,AO2,Miyazaki2022_AOspecial,Oosawa2021}.
This Asakura--Oosawa (AO) model, which is also
called Asakura--Oosawa--Vrij model \cite{VRIJ},
involves a binary mixture of a big particle species, called ``colloids'' (c), and a small one, called ``polymers'' (p) or ``depletants'', which interact solely by excluded volumes,
specified by the interaction diameter $\sigma_{ij}\geq0$, i.e., the minimal distance that a particle of species $i$ can approach a particle of species $j$, where $\{i,j\}=\{\text{p},\text{c}\}$.
The small particles are ideal, so their interaction diameter $\sigma_\text{pp}=0$ is vanishing, where we used the species index p for the small particles. Conversely the big particles
(with species index c) interact among themselves through the nonzero interaction diameter $\sigma_\text{cc}>0$
and are therefore in common physical terms referred to as hard spheres.
The cross-interaction between small and big particles is again hard with a non-additive interaction diameter
 $\sigma_\text{pc}>(\sigma_\text{pp}+\sigma_\text{cc})/2=\sigma_\text{cc}/2$. 
  This inequality defines the present model as a mixture which exhibits a positive non-additivity.
Or in other words, around the big particles there is a spherical excluded volume region which is depleted by small particles
as the probability to stay there is zero.

The AO model constitutes a paradigm for coarse-graining
a binary mixture towards an effective one-component system by integrating out the degrees of freedom of the small particles
\cite{Hansen_Lowen,Fortini_2006}. 
This culminates by the important concept of depletion attraction between the colloids (the big particles) as induced by the
osmotic pressure of the non-adsorbing polymer coils (the small particles) \cite{PuseyLes_Houches_1991,Tuinier_Lekkerkerker_textbook}. For a pair of colloidal
particles this effective depletion attraction can be calculated analytically and equals the overlap volume
of the two spherical depletion regions around the colloids times the osmotic polymer pressure.
Taking the inverse route, this coarse graining scheme also allows to efficiently treat pairwise attractions in a one-component system  by evaluating a (usually more accurate) theory for mixtures \cite{Maeritz2021a,Maeritz2021b}.

Another remarkable aspect of the
AO model that has been established in the physics literature \cite{Gast1983,Evans1999,Evans2003,Hansen_Lowen}
is that for  $\sigma_\text{pc}/\sigma_\text{cc} < 1/\sqrt{3}$ the depletion zones of any statistically
possible configuration of big particles {\it never\/}  has a triple (or higher order) overlap.
This implies that the effective triplet and higher-order effective interactions between the big particles \cite{Allahyarov,monch0jorda2003,Gruenberg,Santos2015,Kobayashi2019,Campos2021} do vanish when integrating out the small particles.
This important statistical feature  shows that there exists an exact mapping onto
an effective one-component system of big particles which then only interact via effective depletion
pair-interactions if the depletion zone is sufficiently small.
Hence the full AO system can be effectively  viewed as a pairwise system interacting via the
 hard-core repulsion of the colloids plus attractive tail potential defined by the depletion interaction due to the presence of ``polymers''.
 For these systems,  standard
liquid state theory can be applied to obtain static correlation functions and equilibrium
phase transitions \cite{Hansen_MacDonald_book}.
 Despite the fundamental importance of depletion interactions in physics, the criteria commonly used to demonstrate their exactness have not been stated in a rigorous mathematical way.

The aim of this paper is threefold.
First, we prove that the  geometric criterion $\sigma_\text{pc}/\sigma_\text{cc} < 1/\sqrt{3}$ for the absence of triplet interactions 
in the standard AO model with identical and spherical colloids is both necessary and sufficient in a strict mathematical sense.  
Second, we focus on more generalized conditions needed to guarantee the absence of triplet and higher-order effective interactions.
By using pure geometric concepts, we extend the geometric condition to polydisperse hard-sphere mixtures and prove sufficient criteria for anisotropic shapes
of the depleted particles \cite{glaser2015parallel,Wood2021,Peters2021,Mason2021,Cheng2021,Calero2021,Santos2021,kuhnhold2022structure}.
Although statistical properties of polydisperse mixtures 
and mixtures of hard particles with different shapes were considered
to a large extent, the basic question of an exact mapping onto a pure pairwise model was not yet much
addressed in the literature for these general situations.
Third, we consider the AO mixture in contact with a planar external hard wall,
previously considered for wetting and layering situations, see, e.g., Refs.~\cite{Widom,Brader2001,interface_AO_1,interface_AO_2,interface_AO_3,interface_AO_4,review_AO}.
In this case we provide a geometric criterion which establishes the exact mapping on
an effective one-component system with one-body external interactions only.

 This paper is structured as follows. In Sec.~\ref{sec_AOmodel} we introduce the AO model in physical terms and briefly state the conditions for the absence of triplet depletion interactions known in the physics literature.
 In Sec.~\ref{sec_proofMAIN} we switch to a mathematical language and state our main theorems, providing a rigorous and generalized set of conditions for the absence of triplet depletion interactions.
The proof of these theorems is completed in Sec.~\ref{sec_proofPROP} by making contact with the historical Apollonius problem and Descartes' circle theorem.
Finally, in Sec.~\ref{sec_conclusions}, we restate the implications of our main theorem in physical terms and make some concluding remarks.

\begin{figure}[t!]
    \centering
       \includegraphics[width=\textwidth]{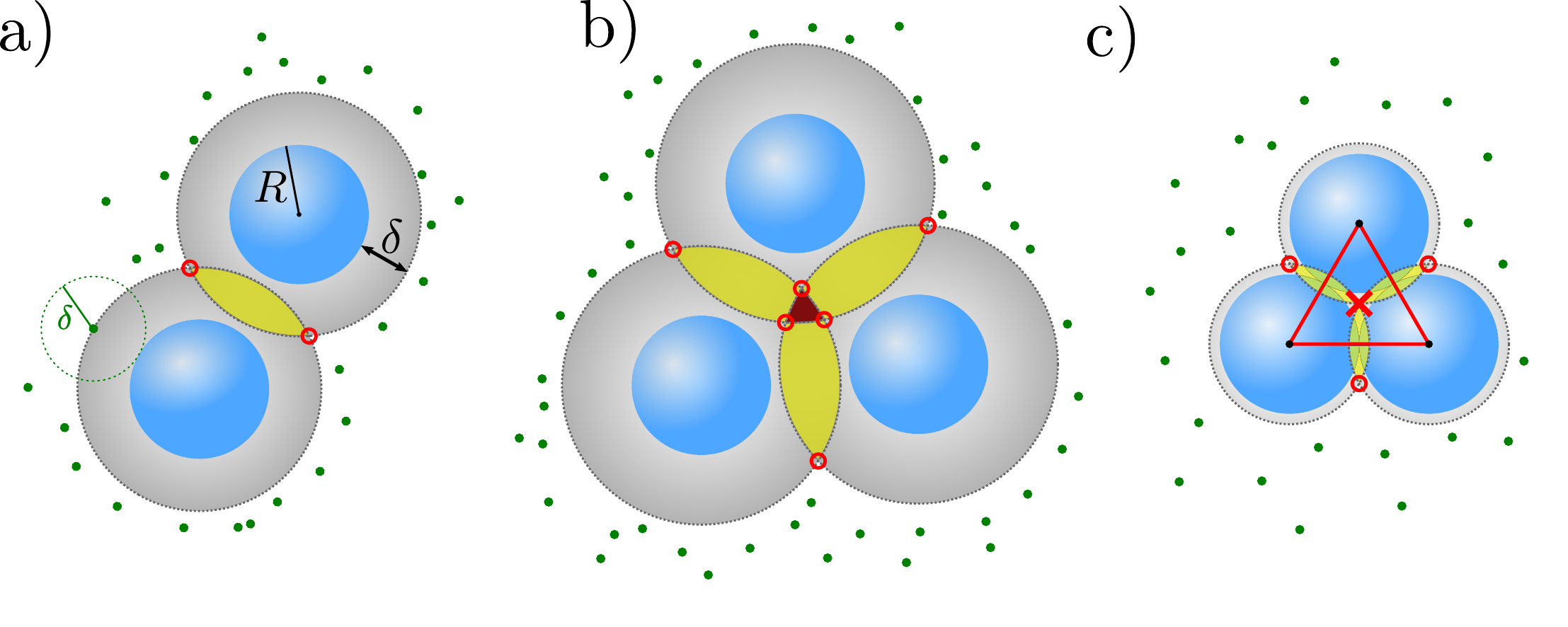}
    \caption{
    Schematic of colloids modeled as large hard spheres (blue) of radius $R$, Eq.~\eqref{eq_Vcc}, in the presence of polymers (or depletants), modeled as ideal point particles (green dots), Eq.~\eqref{eq_Vpp}, with a depletion radius $\delta$ 
    such that the centers of the depletants are excluded from the depletion shells (gray) of radius $R+\delta$ around each colloid, Eq.~\eqref{eq_Vpc}.
    For the purpose of illustration, we draw a dashed green depletion circle around one depletant which is at the smallest possible distance to a colloid.
    The free volume for the centers of the polymers can be increased when the depletion shells of the colloids intersect (an overlap of the hard cores of the colloids is forbidden).    
   Here we highlight the regions of pair overlap (yellow) and triple overlap (red) of the depletion shells and the points of pair overlap (red circles) 
   and triple overlap (red cross) of their surfaces for different configurations and depletion radii $\delta$.
    \textbf{a)} Pairwise intersection of two depletion shells. The volume of the overlap region gives rise to the effective depletion pair potential $V_{\mathrm{cc}}^\text{eff}(r)$ 
    between the colloids, given by Eq.~\eqref{eq_depletionpotential}.
    \textbf{b)} Triple intersection of three depletion shells. The corresponding effective depletion triplet potential is not exactly known.  
    \textbf{c)} Critical configuration with three spheres at mutual contact and maximal depletion radius $\delta=\delta_\text{cr}$, such that the depletion shells overlap in a single point. 
    The value of $\delta_\text{cr}$ can be easily inferred from the drawn equilateral triangle $\mathcal{T}$ (red) and leads to the condition~\eqref{eq_deltaMO} for the absence of triplet interactions,
    such that the description of the colloid--colloid interaction by $V_{\mathrm{cc}}^\text{eff}(r)$ becomes exact.
    }\label{fig:depletion_figure_3_poly}
\end{figure}

\section{The Asakura--Oosawa (AO) Model \label{sec_AOmodel}}

In its original version, the Asakura-Oosawa (AO) model for colloid--polymer mixtures
is defined for monodisperse spherical colloids in three spatial dimensions, modeled as perfectly hard particles and
ideal (i.e., non-interacting) polymers, also called depletants, see Fig.~\ref{fig:depletion_figure_3_poly} for an illustration of the fundamental aspects. 
Expressing the interaction diameter $\sigma_{\mathrm{cc}}=2R$ in terms of  a spherical radius,
 the hard-core colloid--colloid interaction is
given by the pair potential
\begin{align}
  V_{\mathrm{cc}}(r)=\left\{
    \begin{array}{ll}
      0 &\textrm{if $r\geq 2R$} \\
      \infty&\textrm{if $r < 2R$} 
    \end{array}\right.,
    \label{eq_Vcc}
\end{align}
where $r$ is the distance between two colloid centers.
The colloidal interaction radius $R$ sets a typical length scale, such that there is no geometric overlap between the two spheres for $r<2R$.
The infinite pair potential implies a vanishing Boltzmann factor $\exp (- V_\mathrm{cc} (r)/(k_\text{B}T))$, where $k_\text{B}$ is the Boltzmann constant and $T$ denotes the temperature,
in the classical partition sum so that the statistical weight is zero for any overlapping configuration of these hard spheres.
In other words, the particles will not overlap.
The polymer--polymer interaction potential
\begin{align}
  V_{\mathrm{pp}}(r)=0\,,
    \label{eq_Vpp}
\end{align}
corresponding to the interaction diameter $\sigma_{\mathrm{pp}}=0$, vanishes, which means that the polymers are treated as point particles.
The essential cross-interaction between the colloids and polymers is pairwise and given by the non-additive interaction diameter $\sigma_{\mathrm{pc}}=R+\delta$, 
which translates to the hard-core pair interaction potential
\begin{align}
  V_{\mathrm{cp}}(r)=V_{\mathrm{pc}}(r)= \left\{
    \begin{array}{ll}
      0 &\textrm{if $r\geq R+\delta $} \\
      \infty&\textrm{if
      $r < R+\delta$} 
    \end{array}\right.
    \label{eq_Vpc}
\end{align}
with $r$ now denoting the distance between a colloid and polymer center. This interaction
introduces the radius $\delta>0$ of the depletion shell, which is the minimal distance a polymer can
approach the  spherical surface of the colloid.
The case $\delta=0$ corresponds to an additive mixture without depletion interactions.

The crucial idea behind the AO model is that one can integrate out the degrees of freedom associated with the polymer coordinates
and then consider an effective colloid--colloid interaction.
This depletion interaction reflects the increase in the free volume accessible to the polymer centers, 
which is the overall system volume minus the total depletion zone arising from all colloidal particles in the system,
if the depletion shells of nearby colloids overlap.
Considering the lens-shaped overlap region of two spherical depletion shells in three spatial dimensions, as illustrated in Fig.~\ref{fig:depletion_figure_3_poly}a,
it can be  shown exactly that the effective depletion pair potential between two colloids with center-to-center distance $r$ reads
\begin{align}\label{eq_depletionpotential}
V_{\mathrm{cc}}^\text{eff}(r)=V_{\mathrm{cc}}(r)+ V_{\mathrm{cc}}^\text{dep}(r)
\end{align}
with the negative contribution
\begin{align}
 V_{\mathrm{cc}}^\text{dep}(r)=
\begin{cases}
  - P_\mathrm{p} \frac{4\pi (R+\delta)^3}{3}
  \left(1 - \frac{3r}{4(R+\delta)} +\frac{r^3}{16(R+\delta)^3}\right) & \text{if } 2R \leq r \leq 2R+2\delta\\
 \ \ \ \ \ \ \ \ \ \ \ \ \ \ \ \ \ \ \ \ \ \ \ \ \ \ \ \  0 & \text{else } 
\end{cases}
\end{align}
due to depletion if two colloids are sufficiently close such that their depletion shells overlap.
The strength of this depletion interaction depends on the osmotic pressure $P_\mathrm{p}=k_\text{B}T\rho_\mathrm{p}$ of the polymers at density $\rho_\mathrm{p}$.

For a triplet configuration of spheres, as shown in Fig.~\ref{fig:depletion_figure_3_poly}b, the essential
point is how exactly the individual depletion shells, which are inaccessible for the polymers, overlap.
If there is a region of triple intersection of three depletion shells, then the volume of the overall depletion zone cannot be calculated from the individual depletion shells and their pair overlaps alone.
In the language of statistical mechanics~\cite{Hansen_MacDonald_book} this implies
that effective many-body interaction between the colloids are arising in addition to the effective pair potential $V_{\mathrm{cc}}^\text{eff}(r)$. 
In turn, if for {\it any} configuration of three and more colloids there exists no point where all of their depletion shells do intersect, then triplet and higher many-body interactions do vanish.
This reduces the condition for the absence of triplet interactions to a pure geometric overlap problem.
 In the special situation of monodisperse hard spheres of the same radius $R$, as considered so far, 
 the geometric criterion for the depletion radius  $\delta=2\sigma_\text{pc}R/\sigma_\text{cc}-R$  reads \cite{Gast1983,Evans1999,Evans2003,Hansen_Lowen}
\begin{align}
\frac{\delta}{R} < \frac{2}{\sqrt{3}} -1\,,
\label{eq_deltaMO}
\end{align}
so that a triple overlap of the depletion shells is excluded for any configuration of three spheres.
It can be inferred from the configuration, shown in Fig.~\ref{fig:depletion_figure_3_poly}c, where the three spheres are at mutual contact.
Moreover, if the spheres are at contact with a hard planar external wall,
the criterion for the absence of pairwise depletion interactions between the colloids and the wall reads \cite{Brader2001,Evans2003}
\begin{align}
\frac{\delta}{R} < \frac{1}{4}\,.
\label{eq_deltaPL}
\end{align}
The importance of the bound from Eq.~\eqref{eq_deltaMO}, which implies the condition in Eq.~\eqref{eq_deltaPL}, stems from the exactness of the AO model, i.e.,
describing the colloidal degrees of freedom only by the pairwise effective depletion potential, Eq.~\eqref{eq_depletionpotential}, if $\delta$ is below the given threshold.

 \begin{figure}[t!]
    \centering
       \includegraphics[width=\textwidth]{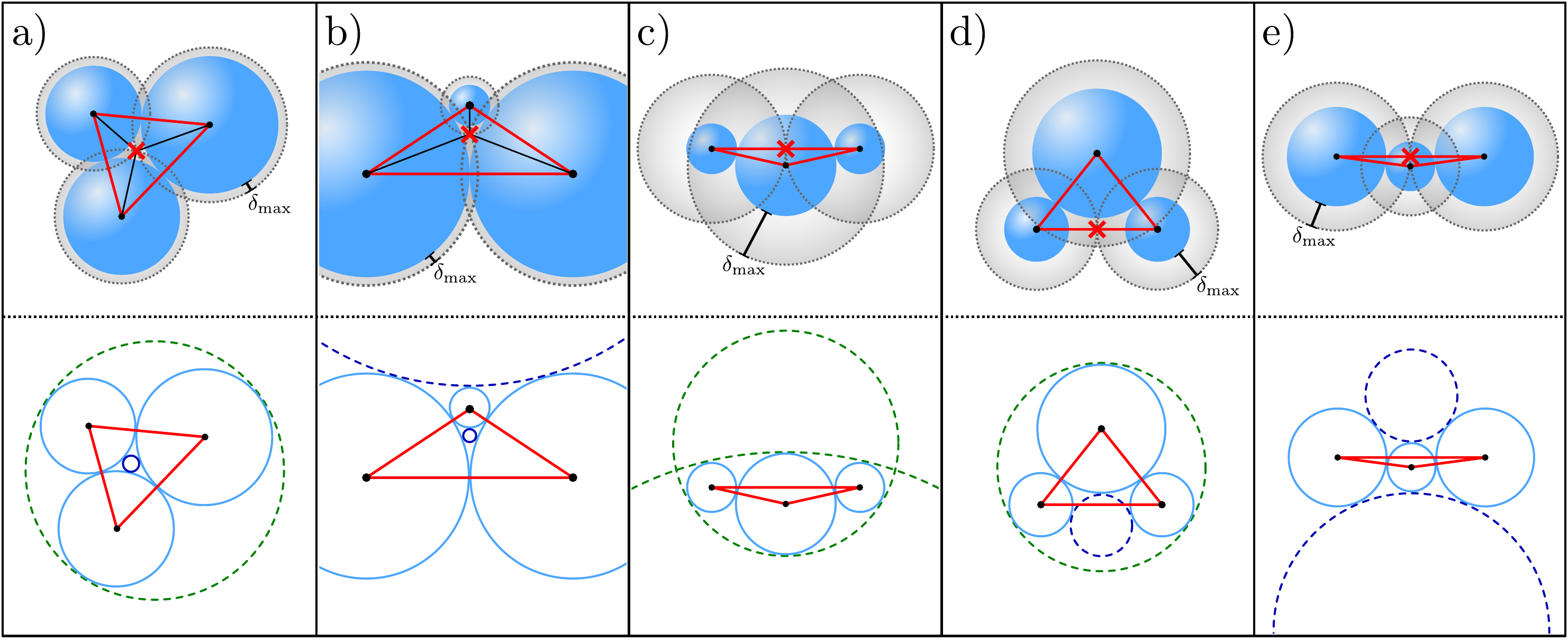}
       \caption{
Illustration of how to determine the maximal depletion radius  $\delta_\text{max}>0$ around 
three balls with different radii $R_i$, such that there is no triple intersection for $\delta<\delta_\text{max}$.
We show a projection of five representative configurations onto a plane  containing the triangle $\mathcal T$ (red) whose corners are centers of the balls (more details are given in Fig.~\ref{fig:circle3_arb}).
The balls can thus be treated as two-dimensional disks. 
 \textbf{Top panel:} location of the single point $\mathbf p$ of triple intersection (red cross) of the depletion zones (gray) of the three disks (blue) with total radii $R_i+\delta_\text{max}$, compare Prop.~\ref{prop:dichotomy}.
\textbf{Bottom panel:} (segments of) solution circles of the related Apollonius problem, posed in the specific form to find a circle (dark blue) that is externally tangent to the three given circles (which are mutually externally tangent).
Internally tangent solution circles (dark green) have no relevance for the depletion problem.
 The center of a solution circle represents the point $\mathbf p$ (and its radius $\delta^*=\delta_\text{max}$ equals the maximal depletion radius)
if and only if it is located inside the triangle $\mathcal T$ (red) and the circle is externally tangent, compare Prop.~\ref{prop:dichotomy}.
These circles are drawn with solid lines, all other solutions with dashed lines.
The columns depict a representative scenario with
       \textbf{a)} one externally tangent solution circle with center inside $\mathcal T$;
       \textbf{b)} two externally tangent solutions, exactly one circle has its center inside $\mathcal T$;
       \textbf{c)} no externally tangent solutions;
       \textbf{d)} one  externally tangent solution circle with center outside $\mathcal T$;
       \textbf{e)} two externally tangent solution circles with both centers outside $\mathcal T$.
       \label{fig:circle3}
       }
\end{figure}

This basic AO model for monodisperse hard spheres can be generalized towards hard hyperspheres (balls) in arbitrary spatial dimensions $n\in \mathbb{N}$ in a straightforward manner by taking for $r$ the Euclidean distance in arbitrary dimensions. 
Moreover, it can be generalized to arbitrarily shaped hard colloidal bodies in the following way:
 if -- for a given configuration of colloids -- the volumes of the two bodies overlap, the pair potential $V_\mathrm{cc}$ is formally infinite so that overlap is excluded. 
The polymer--colloid interaction is infinity if the polymer center is closer than the depletion distance $\delta$ to at least one point of the colloidal surfaces.
 In three spatial dimensions, this generalization includes polydisperse colloidal hard spheres with different radii $R_i$  and therefore
 also a planar hard wall in the limit where one of these radii diverges. 
 It also includes nonspherical shapes, both convex and nonconvex.

We are particularly interested in the following questions.
Why is the arrangement of three hard spheres depicted in Fig.~\ref{fig:depletion_figure_3_poly}c the critical one, which gives rise to the upper bound for the depletion radius stated in Eq.~\eqref{eq_deltaMO}?
What is the critical configuration if we consider mixtures with different radii?
How must Eq.~\eqref{eq_deltaMO} be generalized in such a case, or in general, for arbitrary mixtures of different hard bodies with an anisotropic shape?
To answer these questions in a mathematically rigorous way, we extend the problem and consider an arbitrary but fixed arrangement of three balls, for which we seek the maximal shell radius $\delta_\text{max}$, 
such that the triple overlap of the depletion shells is exactly in one point, which implies the general condition $\delta<\delta_\text{max}$.
As illustrated in Fig.~\ref{fig:circle3}, the definition of $\delta_\text{max}$ depends on the configuration.
 In what follows, we define $\delta_\text{max}$ in a formal way and analyze its configurational dependence in mathematical terms to determine its minimal value among all possible configurations, 
which provides both generalized conditions for the absence of triplet interactions and a rigorous proof thereof.

\section{Conditions for the absence of triplet interactions \label{sec_proofMAIN}}

Our main mathematical result is a sufficient condition for the absence of triple overlap of dilated versions of non-overlapping convex bodies in $\mathbb R^n$ with $n\geq 2$. By ``non-overlapping'' we mean that the interiors of the convex bodies are disjoint; contact in boundary points is allowed.  Given $\delta>0$ and $C\subset \mathbb R^n$,  consider the dilated set (often called \emph{outer parallel body}) given by 
\begin{equation} 
	C(\delta) = \bigl \{ \mathbf{r} \in \mathbb R^n:\, \mathrm{dist}(\mathbf{r}, C) \leq \delta\bigr \},
\end{equation} 	
where the distance of a point to a set is defined in terms of the Euclidean distance $|\cdot|$ as $\mathrm{dist}(\mathbf{r}, C) = \inf_{\mathbf{y} \in C} |\mathbf{y}- \mathbf{r}|$. 

The sufficient condition for absence of triple overlap involves a bound on the rolling radius of the (non-dilated) sets. A closed convex set $C$ has positive rolling radius if a ball can roll freely inside $C$ along the boundary,
i.e., if for some radius $r>0$ and all $\mathbf y\in \partial C$ there exists a closed ball $B(\mathbf x, r)$ with $\mathbf y\in B(\mathbf x, r)\subset C$.
The \emph{rolling radius} $\mathrm{Roll}(C)$ is the supremum of the possible radii $r>0$. If no ball with radius $r>0$ can roll freely inside $C$, we assign $\mathrm{Roll}(C)=0$.

\begin{remark}
When $\partial C$ is a $C^2$ surface, Blaschke's rolling theorems \cite{blaschke56} (see also Ref.~\cite{howard99} for extensions and additional references) show that the rolling radius is the minimum of the inverse of the principal curvatures. When smoothness of the boundary $\partial C$ is not assumed, curvature in convex geometry can be defined in terms of curvature measures. Bounded curvature is replaced with absolute continuity of the curvature measure with respect to a surface measure and the condition that the Radon-Nikodym derivative is finite; a relation with strictly positive rolling radius can be established in this general context as well \cite{bangert99,hug99,hug2002}.
\end{remark} 

\begin{theorem} \label{thm:main1} 
Fix $n\geq 2$ and $\kappa \geq 3$.
	Let $C_i$, $i=1,2,\ldots,\kappa$ be compact convex subsets of $\mathbb R^n$ with disjoint interiors. Suppose that they have positive rolling radii $\mathrm{Roll}(C_i)>0$ and that 
	\begin{equation} 
		\delta < \Bigl(\frac{2}{\sqrt 3} - 1\Bigr) \min_{i=1,\ldots,\kappa} \mathrm{Roll}(C_i)\,.
		\label{eq_deltaMOmathGEN}
	\end{equation} 
	Then the intersection of the dilated sets is empty: 
	$\bigcap_{i=1}^\kappa C_i(\delta) =  \varnothing$.  If instead 
	\begin{equation} 
		\delta \leq \Bigl(\frac{2}{\sqrt 3} - 1\Bigr) \min_{i=1,\ldots,\kappa} \mathrm{Roll}(C_i)\,
		\label{eq:suff2}
	\end{equation} 
	then the intersection has empty interior and zero Lebesgue measure: $\mathrm{vol}(\bigcap_{i=1}^\kappa C_i(\delta)) = 0$. 
\end{theorem} 

 Theorem~\ref{thm:main1} is complemented at the end of this section by a result when one of the $C_i$'s is a half-plane (``hard wall''), see Theorem~\ref{thm:main1WALL}. 

In the situation of condition~\eqref{eq:suff2} we shall say that \emph{there is no overlap} or that \emph{there is no triplet (or higher-order) interaction}. This is because effective interactions are given by volumes of intersections and may vanish even if the intersection is non-empty, e.g., when it consists of a single point.

Theorem~\ref{thm:main1} is a consequence of the following theorem on triple overlap of dilation of closed balls,  which also provides a necessary criterion in this special case.

\begin{theorem} \label{thm:main2}
	Fix $n\geq 2$ and $R>0$.
	The following two conditions are equivalent:
	\begin{enumerate} 
		\item [(i)] $\delta>0$ and $R>0$ satisfy 
		\begin{equation} 
			\delta < \Bigl(\frac{2}{\sqrt 3} - 1\Bigr)R. 
		\label{eq_deltaMOmath}
		\end{equation}
		\item [(ii)] For all closed balls $B(\mathbf{r_i}, R_i)\subset \mathbb R^n$, $i=1,2,3$ with $R_i>0$ and centers $\mathbf{r}_i\in \mathbb R^n$ that have disjoint interiors and satisfy 
		$\min(R_1,R_2,R_3) \geq R$, 
		the dilated balls $B(\mathbf{r}_i, R_i+\delta)$ have empty triple intersection.
	\end{enumerate} 
\end{theorem}

Before we address the proof of Theorem~\ref{thm:main2}, let us show how Theorem~\ref{thm:main1} follows.
In fact, it is enough to know  the implication (i) $\Rightarrow$ (ii) from Theorem~\ref{thm:main2} for balls that have the same radius $R$.

\begin{proof} [Proof of Theorem~\ref{thm:main1}]
	Consider first triple intersection, $\kappa =3$. We show that if the triple overlap is non-empty, then the condition on $\delta$ is violated. Thus, suppose that there exists a point $\mathbf{r}\in \mathbb R^n$ that is in the triple intersection of the dilations $C_i(\delta)$, $i=1,2,3$. 
	 Let $\mathbf{p}_i$ be the projection of $\mathbf{r}$ onto $C_i$, i.e., the uniquely defined point $\mathbf{p}_i\in C_i$ with $|\mathbf{r}- \mathbf{p}_i| = \mathrm{dist}(\mathbf{r}, C_i)\leq \delta$.
	 Set $R_{\min}:= \min_{i=1,2,3} \mathrm{Roll}(C_i)$. If $\mathbf r$ is in none of the interiors of the $C_i$'s, then $\mathbf p_i \in \partial C_i$ and the definition of the rolling radius guarantees the existence of three balls  $B(\mathbf x_i, R_{\min})$ such that $\mathbf p_i \in B(\mathbf x_i, R_{\min}) \subset C_i$ for $i=1,2,3$. The three balls of radius $R_{\min}$ have disjoint interiors and the triple intersection $\cap_{i=1}^3 B(\mathbf x_i, R_{\min}+ \delta)$ is non-empty as it contains $\mathbf r$. Theorem~\ref{thm:main2} with $R_i = R_{\min}$ implies that $\delta \geq (2/\sqrt 3 - 1) R_{\min}$. 

	Now suppose $\mathbf r$ is in the interior of one of the $C_i$'s,  say $C_1$. By the disjointness of the interiors of the sets $C_i$, the point $\mathbf r$ is not in the interior of $C_2$ or $C_3$, and we can define the projections $\mathbf p_i \in \partial C_i$ for $i=2,3$ and the associated  balls $B_i = B(\mathbf x_i, R_{\min})$, $i=2,3$ as before. 
	
	We claim that there is also a closed ball $B'_1\subset C_1$ of radius $R_{\min}$ that contains $\mathbf r$. To see why, consider the set of all balls that contain $\mathbf r$ and are contained in $C_1$, and among them let $B(\mathbf x_1,R)$ be a ball of maximum possible radius. Then $B(\mathbf x_1, R)$ touches the boundary $\partial C_1$ in at least one point $\mathbf p_1$. Working with the boundary point $\mathbf p_1$ and the definition of the rolling radius one sees that $R$ must be larger or equal to the rolling radius $R_1$ of $C_1$, hence $R\geq R_1\geq R_{\min}$. In particular, there exists a ball $B'_1\subset B(\mathbf x_1,R) \subset C_1$ of radius $R_{\min}$ with $\mathbf r\in B'_1$. 
	
	The balls $B'_1$, $B_2$, $B_3$ have disjoint interiors and their $\delta$-dilations contain $\mathbf r$, we conclude again with Theorem~\ref{thm:main2} applied to $R_i = R_{\min}$ that $\delta \geq (2/\sqrt 3 - 1) R_{\min}$.  This completes the proof of empty intersection for $\kappa =3$. The claim for $\kappa \geq 3$ follows right away since the intersection $\bigcap_{i=1}^\kappa C_i(\delta)$ is contained in the triple intersection $\bigcap_{i=1}^3 C_i(\delta)$. 
	
	For the second part of the theorem the only case left to investigate is when $\delta$ is equal to the right-hand side of condition~\eqref{eq:suff2}. A continuity argument shows that the convex set $\cap_{i=1}^\kappa C_i(\delta)$ has empty interior: the interior consists of those points $y\in \mathbb R^n$ that have distance strictly smaller than $\delta$ to each of the sets $C_i$. Every such point is in the intersection $\cap_{i=1}^\kappa C_i(\delta')$ for some $\delta'<\delta$. As the latter intersection is empty by the part of the theorem already proven, no such point exists and the interior is empty.  The only way for a convex set to have empty interior is that the convex set is contained in some affine hyperplane, in which case the measure is zero.
\end{proof}

\begin{remark} \label{rem:convexity}
	 A close look at the proof reveals that convexity and compactness are not essential for Theorem~\ref{thm:main1}:
	If, for some radius $R>0$, each $C_i$ has the property that every point $\mathbf p \in C_i$ is contained in a ball of radius $R$ that lies entirely in $C_i$ (for convex bodies, this implies $R\leq \min_{i=1,\ldots,\kappa} \mathrm{Roll}(C_i)$), and $\delta$ satisfies condition~\eqref{eq_deltaMOmath} for this $R$, then the dilated bodies $C_i(\delta)$ have empty intersection.
\end{remark}

\begin{remark} \label{rem:improvedcriterion}
The bound in Theorem~\ref{thm:main1} can be improved for some compact convex sets with nonempty interior, specifically for the case of convex polyhedra with $\mathrm{Roll}(C_i)=0$. 
Choose for each convex and compact $C_i\subset\mathbb R^n$, $i=1,2,\ldots,\kappa$ with  disjoint interiors
	a convex and compact subset $K_i\subseteq C_i$ with $\mathrm{Roll}(K_i)>0$.
	By definition, all $K_i$ have disjoint interiors.
	Then, it follows by Theorem~\ref{thm:main1} that the intersection of the dilated sets $K_i(\Delta)$ is empty if the pseudo-depletion radius $\Delta$ satisfies $\Delta<\min_{i=1,\ldots,\kappa}\Delta_{K_i}$
	with
	\begin{equation}
	\Delta_{K_i}:=\Bigl(\frac{2}{\sqrt 3} - 1\Bigr) \mathrm{Roll}(K_i)\,.
	\end{equation}
	If there exists a depletion radius $\delta'$ with $0\leq\delta'\leq\min_{i=1,\ldots,\kappa}\Delta_{K_i}$ such that $ C_i(\delta')\subseteq K_i(\Delta_{K_i})$ for all $i=1,2,\ldots,\kappa$,
	then also the intersection of all dilated sets $C_i(\delta')$ is empty.
	If such a $\delta'$ exceeds the bound in Eq.~\eqref{eq_deltaMOmathGEN} we have found an improvement.
	
	To properly define a new bound that is, at least in a restricted sense, optimal, let us first define the optimal radius
	$\Delta^{\!*}(K_i):=
	\sup\{\delta'\geq 0:\ C_i(\delta')\subseteq K_i(\Delta_{K_i})\}$ for any $K_i$ and subsequently
	$\Delta^{\!*}_i:=\sup_{K_i\subseteq C_i}\Delta^{\!*}(K_i)$, i.e., optimizing over all convex and compact subsets.
	Then the condition
	\begin{equation}
	\delta<\min_{i=1,\ldots,\kappa}\Delta^{\!*}_i=\min_{i=1,\ldots,\kappa}\;\sup_{K_i\subseteq C_i}\;
	\sup\{\delta'\geq 0:\ C_i(\delta')\subseteq K_i(\Delta_{K_i})\}
	\end{equation}
 implies that $\bigcap_{i=1}^\kappa C_i(\delta) =  \varnothing$.
It is obvious that this bound is better or equal to that in Eq.~\eqref{eq_deltaMOmathGEN}, since the bounds are equal if the
supremum is attained for $K_i=C_i$ in case of the $C_i$ with the smallest rolling radius.
It remains an open problem whether this extended criterion is now also necessary in the sense of Theorem~\ref{thm:main2}.
\end{remark} 

For the proof of Theorem~\ref{thm:main2}, we start from three non-overlapping hard bodies and grow the shell radius $\delta$ from $\delta =0$ until the triple intersection first becomes non-empty. 
This defines a threshold value $\delta_{\max}$ that we compute explicitly when the three bodies are balls in contact.
The threshold value is defined precisely in the next lemma, which works for general compact bodies.

\begin{lemma}[Existence and uniqueness of maximal shell radius for triple intersection]
	\label{lem_deltamaxUNIQUE}
	Let $C_i$, $i=1,2,3$ be three non-empty compact subsets of $\mathbb R^n$ with disjoint interiors. 
	Then there exists a uniquely defined $\delta_{\max} = \delta_{\max}(C_1,C_2,C_3) \in [0,\infty)$ such
	that the triple intersection of the dilated sets $C_i(\delta)$ is empty if and only if $\delta <\delta_{\max}$:
	\begin{equation}\label{eq:deltamax-def}
		C_1(\delta) \cap C_2(\delta) \cap C_3(\delta) = \varnothing\ \Leftrightarrow \delta < \delta_{\max}(C_1,C_2,C_3). 
	\end{equation}
\end{lemma} 

\begin{proof}
	The threshold value $\delta_{\max}$ is clearly unique. It remains to prove existence. 
	Let 
	\begin{equation} \label{eq:deltamax-def2}
		\delta_{\max}:= \sup\bigl\{ \delta>0:\ C_1(\delta) \cap C_2(\delta) \cap C_3(\delta) = \varnothing\bigr\}
	\end{equation}  
	if the latter set is non-empty, and $\delta_{\max} =0$ otherwise (this may occur when the convex bodies touch in a common cusp of the surfaces $\partial C_i$).
	Clearly  the triple intersection is non-empty for all sufficiently large $\delta$, therefore $\delta_{\max}<\infty$. 
	
	To check that Eq.~\eqref{eq:deltamax-def} holds true, we note that the map $\delta \mapsto \cap_{i=1}^3 C_i(\delta) =: \mathcal I(\delta)$ is monotone increasing, 
	i.e., $\delta\leq \delta'$ implies $\mathcal I(\delta) \subset \mathcal I(\delta')$. 
	Let $\delta <\delta_{\max}$. Then by the definition~\eqref{eq:deltamax-def2} of $\delta_{\max}$ there exists $\delta'$ with $\delta \leq \delta' \leq \delta_{\max}$, 
	such that $\mathcal I(\delta') = \varnothing$. The inclusion $\mathcal I(\delta)\subset \mathcal I(\delta')$ then implies $\mathcal I(\delta) = \varnothing$. 
	For $\delta>\delta_{\max}$, we get right away from~\eqref{eq:deltamax-def2} that the triple intersection is non-empty. 
	
	It remains to check that at $\delta = \delta_{\max}$, the triple intersection $\mathcal I(\delta)$ is non-empty. To that aim let $(\delta_n)_{n\in \mathbb N}$ be a strictly decreasing sequence with $\delta_n\downarrow \delta_{\max}$.  Let $(\mathbf x_n)_{n\in \mathbb N}$ be a sequence of points with $\mathbf x_n\in  \mathcal I(\delta_n)$. The sequence $(\mathbf x_n)$ is bounded and therefore admits an accumulation point $\mathbf x^*\in \mathbb R^n$; assume for simplicity that $\mathbf x_n\to \mathbf x^*$ (if not, pass to a subsequence).  By the continuity of the maps $\mathbf x\mapsto \mathrm{dist}(\mathbf x, C_i)$, the limit point $\mathbf x^*$ satisfies 
	\begin{equation}
		\mathrm{dist}(\mathbf x^*, C_i) = \lim_{n\to \infty} \mathrm{dist}(\mathbf x_n, C_i) \leq \lim_{n\to \infty} \delta_n = \delta_{\max}.
	\end{equation} 
	for $i=1,2,3$. As a consequence, $\mathbf x^*$ is in $\mathcal I(\delta_{\max})$, hence $\mathcal I(\delta_{\max})$ is non-empty. It follows that $\delta_{\max}$ defined in~\eqref{eq:deltamax-def2} satisfies~\eqref{eq:deltamax-def}.
\end{proof} 

Next we specialize to convex bodies that are closed balls. For balls it is enough to understand the two-dimensional setup, as depicted in Fig.~\ref{fig:circle3}.

\begin{lemma} [Reduced dimensionality of the overlap problem] \label{lem:reduced}
	Let $B_i = B(\mathbf r_i, R_i)$, $i=1,2,3$ be three  closed balls with disjoint interiors and $\mathcal P \subset \mathbb R^n$ a two-dimensional affine subspace containing the three centers $\mathbf r_i$, $i=1,2,3$. Then for all $\delta>0$, 
	\begin{equation} 
		\bigcap_{i=1}^3 B(\mathbf r_i,R_i+\delta) = \varnothing \ \Leftrightarrow\ \mathcal P\cap \Bigl(\bigcap_{i=1}^3 B(\mathbf r_i,R_i+\delta)\Bigr) = \varnothing. 
	\end{equation} 
\end{lemma} 

\begin{proof}
	The implication ``$\Rightarrow$'' is trivial. The converse implication follows once we prove 
	\begin{equation} \label{eq:reduced-dim}
		\bigcap_{i=1}^3 B(\mathbf r_i,R_i+\delta) \neq \varnothing \ \Rightarrow\ \mathcal P\cap \Bigl(\bigcap_{i=1}^3 B(\mathbf r_i,R_i+\delta)\Bigr) \neq \varnothing.
	\end{equation} 
	Let $\mathbf r\in \bigcap_{i=1}^3 B(\mathbf r_i,R_i+\delta)$ and $\mathbf {r}^\perp$ the orthogonal projection of $\mathbf r$ onto the plane $\mathcal P$. Then $|\mathbf r_i - \mathbf r^\perp| \leq |\mathbf r_i - \mathbf r| \leq R_i +\delta$, for all $i=1,2,3$. It follows that $\mathbf r^\perp$ is in $\mathcal P \cap (\bigcap_{i=1}^3 B(\mathbf r_i,R_i+\delta))$, in particular the latter set is non-empty.  
\end{proof}

\begin{figure}[t!]
    \centering
        \includegraphics[width=0.8\textwidth]{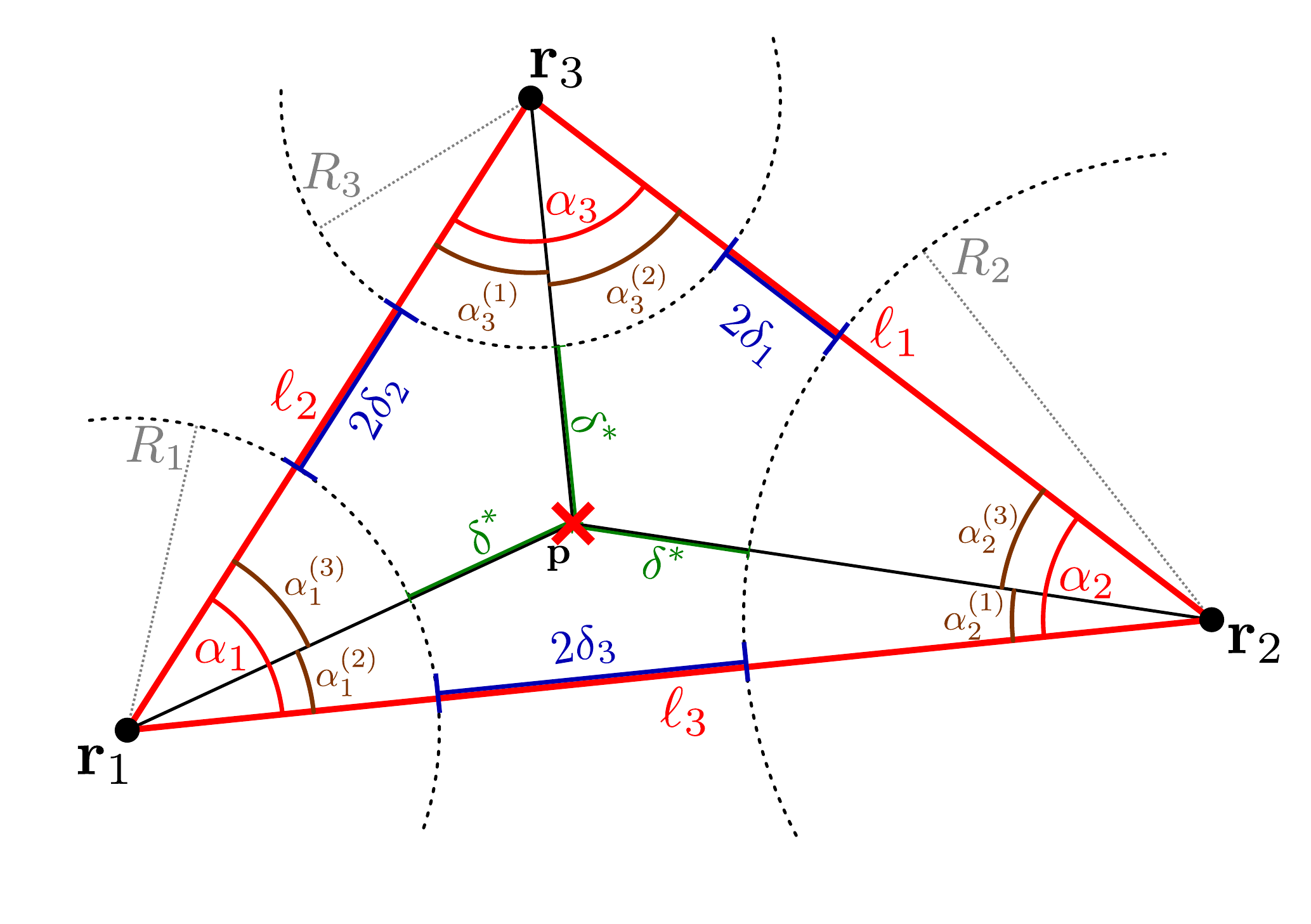} 
    \caption{Geometry of the triangle $\mathcal T$ with edge lengths $\ell_i$, interior angles $\alpha_i$ and corners at $\mathbf{r}_i$ with $i=1,2,3$, as introduced before Prop.~\ref{theorem_da}.  
 The corners correspond to the centers of three disks $B(\mathbf r_i, R_i)$  of radii $R_i$, compare Fig.~\ref{fig:circle3}.
    We also indicate the particular depletion radii $\delta_i$, defined in Eq.~\eqref{eq_deltai}, such that the two dilated disks $B(\mathbf r_j, R_j+\delta_i)$ and $B(\mathbf r_k, R_k+\delta_i)$ are at contact ($\{i,j,k\} = \{1,2,3\}$).
    When the three dilated disks $B(\mathbf r_i, R_i+\delta^*)$ intersect in a single point $\mathbf{p}\in\mathcal T$ (red cross), we define the angles $\alpha_i^{(j)}$ with $j\neq i$ according to Eq.~\eqref{eq:alpha1}.
    }\label{fig:circle3_arb}
\end{figure}

Thus let us focus on dimension $n=2$, in which balls turn into disks. Consider first the non-degenerate case in which the three centers $\mathbf r_i$ of the non-overlapping disks do not lie on a common line.  Let $\mathcal T$ be the triangle with corners $\mathbf r_i$, $i=1,2,3$. The triangle has interior angles $\alpha_i \in (0,\pi)$ and side-lengths $\ell_i$. The indices are such that $\alpha_i$ is the angle at corner $i$ and $\ell_i$ is the length of the triangle side opposite the corner $\mathbf r_i$, see Fig.~\ref{fig:circle3_arb}. Non-overlapping of the balls  is equivalent to the condition $\ell_i \geq R_j+ R_k$ for all pairwise distinct $i,j,k$. 

Given three disks $B_i=B(\mathbf r_i,R_i)$, consider  the maximal shell radius $\delta_{\max}= \delta_{\max}(B_1,B_2,B_3)$ from Lemma~\ref{lem_deltamaxUNIQUE}. At fixed radii $R_1,R_2,R_3$, it is a function of the triangle $\mathcal T$ with corners $\mathbf r_i$ that is invariant under Euclidean isometries of the plane. Up to Euclidean transformations, the triangle $\mathcal T$ is uniquely determined by fixing an angle $\alpha_i$ and the side lengths $\ell_j$, $\ell_k$ of the adjacent triangle sides. Thus we may view $\delta_{\max}$ as a function of, say, $\alpha_1$, $\ell_2$, and $\ell_3$. By some abuse of notation we use the same letter $\delta_{\max}(\alpha_i,\ell_j,\ell_k)$. 

\begin{proposition}[Monotonicity of the maximal shell radius] \label{theorem_da}
	Fix  $R_1,R_2,R_3>0$. Then, for all enumerations $(i,j,k)$ of $\{1,2,3\}$ and all $\ell_j,\ell_k>0$, the map 
	\[
		(0,\pi) \ni \alpha_i\mapsto \delta_{\max}(\alpha_i,\ell_j,\ell_k)
	\] 
	is monotone increasing. 
\end{proposition}

The proof is given in Sec.~\ref{sec_proofPROP}. The proposition says that the decrease in an angle $\alpha_i$, at fixed adjacent side lengths  results in a decrease (or no change at all) of the maximal shell radius $\delta_{\max}$. Now, every triangle satisfying the hard-core constraints $\ell_i \geq R_j+ R_k$, $\{i,j,k\} = \{1,2,3\}$, can be mapped to a triangle in which $\ell_i = R_j+ R_k$, by a succession of two such angle decreasing moves. The only exception, explicitly treated in the proof of Corollary~\ref{cor_deltacr} below, is when an inner angle can be decreased until all centers are collinear, without creating any overlap among the balls.

Let 
\begin{equation} 
	\delta_{\mathrm{cr}}(R_1,R_2,R_3) = \delta_{\max}\bigl( B(\mathbf r'_1,R_1),B(\mathbf r'_2,R_2), B(\mathbf r'_3,R_3)\bigr)
	\label{eq:defdelta-cr}
\end{equation} 
be the critical shell radius in a configuration of non-overlapping disks that are in contact, i.e., 
$|\mathbf r'_i - \mathbf r'_j| = R_i + R_j$ for all distinct $i,j \in \{1,2,3\}$.  Notice that up to Euclidean isometries of the plane, there is a unique such triangle. 

\begin{corollary}[Critical shell radius for fixed radii] \label{cor_deltacr}
	For all radii $R_1,R_2,R_3>0$, and all  disks $B(\mathbf r_i, R_i)\subset \mathbb R^2$, $i=1,2,3$, that have disjoint interiors, we have 
	\begin{equation}
 		\delta_{\max}\bigl(B(\mathbf r_1,R_1),B(\mathbf r_2,R_2), B(\mathbf r_3,R_3)\bigr) \geq 		 \delta_{\mathrm{cr}}(R_1,R_2,R_3)\,.
	\end{equation} 
\end{corollary} 

\begin{proof}
	If the configuration of non-overlapping disks can be turned into a configuration of disks at contact by a sequence of angle-decreasing moves, then the inequality follows from the monotonicity in Prop.~\ref{theorem_da}. 
	If the configuration is instead turned into a configuration of non-overlapping disks with collinear centers, corresponding to a degenerate triangle with one angle $\pi$ and two zero angles, we notice that the monotonicity from Prop.~\ref{theorem_da} extends to the closed interval $[0,\pi]$ by the continuity of the map $\alpha_1\mapsto \delta_{\max}(\alpha_1,\ell_2,\ell_3)$ and therefore the  maximal shell radius $\delta_{\mathrm{max}}$ of the original configuration is still bounded from below by the maximal radius of the aligned disks after decreasing the angle to zero. 
	The latter is easily seen to be larger or equal to the radius of the middle disk, which is clearly larger than the critical radius $\delta_{\mathrm{cr}}(R_1,R_2,R_3)$ by the explicit formula provided in Prop.~\ref{theorem_dR}. Alternatively,  to directly arrive at the critical configuration with all three disks in mutual contact, the maximal shell radius for aligned disks can be made smaller by shrinking distances so that the middle one of the aligned disks is in contact with its two neighbors, and then the shell radius is further decreased by decreasing the angle $\pi$ of the degenerate triangle until we end up with a configuration in which all disks are in contact. 
\end{proof}

The critical shell radius $\delta_{\mathrm{cr}}$ for disks that are in contact is given by an explicit formula.  In Sec.~\ref{sec_proofPROP} we prove that it is equal to the radius of the inner Soddy circle, which can be computed with Descartes' circle theorem, see Ref.~\cite{coxeter68, oldknow96} and references therein.

\begin{proposition} [Critical shell radius at contact] \label{theorem_dR}
	For all $R_1,R_2,R_3$, we have 
	\begin{equation} 
		\delta_{\mathrm{cr}}(R_1,R_2,R_3) = \Bigl( R_1^{-1} + R_2^{-1} + R_3^{-1} + 2 \sqrt{(R_1R_2)^{-1} +(R_2 R_3)^{-1} + (R_1R_3)^{-1}}\Bigr)^{-1}.  
	\end{equation} 
\end{proposition} 

The proposition is proven in Sec.~\ref{sec_proofPROP} by establishing a relation with the Apollonius problem of tangent circles \cite[Chapter VI]{johnson60}. Remember that the latter consists in finding a circle that is tangent to three given non-intersecting circles. Descartes' circle theorem concerns the special case where the three given circles in the Apollonius problem are in contact in three different points. This special case is sometimes called four-coins problem \cite{oldknow96}.  

Prop.~\ref{theorem_dR} implies right away that $\delta_{\mathrm{cr}}$ is strictly increasing in each $R_i$. In particular, setting $R:= \min (R_1,R_2,R_3)$, we get 
\begin{equation} \label{eq:mindelta-cr}
	\delta_{\mathrm{cr}}(R_1,R_2,R_3) \geq \delta_{\mathrm{cr}}(R,R,R) 
	=\Bigl(\frac{2}{\sqrt 3} - 1\Bigr) R.
\end{equation} 
We are now ready to prove Theorem~\ref{thm:main2}. 

\begin{proof}[Proof of Theorem~\ref{thm:main2}] 
	In view of Lemma~\ref{lem:reduced} it is enough to treat the two-dimensional case. Thus let $n=2$. 
	Suppose that $\delta < (\frac{2}{\sqrt 3} -1) \min (R_1,R_2,R_3)$. Then Eq.~\eqref{eq:mindelta-cr} and Corollary~\ref{cor_deltacr} yield
	\begin{equation} 
		\delta < \delta_{\max}\bigl( B(\mathbf r_1, R_1+\delta),B(\mathbf r_2, R_2+\delta), B(\mathbf r_3, R_3+\delta)\bigr).
	\end{equation}
	By the definition of $\delta_{\max}$ in Lemma~\ref{lem_deltamaxUNIQUE}, it follows that the intersection $\cap_{i=1}^3 B(\mathbf r_i, R_i+\delta)$ is empty. This proves the implication (i) $\Rightarrow$ (ii).
	
 Now suppose that $\cap_{i=1}^3 B(\mathbf r_i, R_i+\delta)=\varnothing$ for all positions $\mathbf{r}_i\in \mathbb R^n$.
	In particular, the relation $\cap_{i=1}^3 B(\mathbf r'_i, R+\delta)=\varnothing$ also holds for three balls with radius $R=\min (R_1,R_2,R_3)$ at contact, i.e., 
$|\mathbf r'_i - \mathbf r'_j| = 2R$ for all distinct $i,j \in \{1,2,3\}$. Then the converse implication follows directly from Eq.~\eqref{eq:mindelta-cr} and the uniqueness of $\delta_{\mathrm{cr}}$, defined in Eq.~\eqref{eq:defdelta-cr}, which is implied in Lemma~\ref{lem_deltamaxUNIQUE}.
\end{proof}

We conclude with a sufficient result on absence of triplet intersections for three bodies $C_i$, $i=1,2,3$, when one of them, say $C_3$,  is a half-plane, modeling a hard wall.
A half-plane may be regarded as a ball with infinite radius.  Setting $R:= \min (R_1,R_2)$, we get instead of Eq.~\eqref{eq:mindelta-cr}
\begin{equation} \label{eq:mindelta-crWALL}
	\lim_{R_\text{w}\rightarrow\infty}\delta_{\mathrm{cr}}(R_1,R_2,R_\text{w}) \geq \lim_{R_\text{w}\rightarrow\infty}\delta_{\mathrm{cr}}(R,R,R_\text{w}) 
	=\frac{ R}{4}.
\end{equation} 
which leads to the following theorem. By closed half-space we mean a set of the form $\{ \mathbf x \in \mathbb R^n:\, \mathbf x\cdot \mathbf n \leq c\}$ for some $\mathbf n\in \mathbb R^n$ and some constant $c\in \mathbb R$. 

\begin{theorem} \label{thm:main1WALL} 
Fix $n\geq 2$ and $\kappa \geq 3$.
	Let $C_i$, $i=1,2,\ldots,\kappa-1$ be compact convex subsets of $\mathbb R^n$ and $C_\kappa$ a closed half-space.  Suppose that the $C_i$, $i\leq \kappa -1$ have positive rolling radii $\mathrm{Roll}(C_i)>0$ and that 
	\begin{equation} 
		\delta < \frac{1}{4} \min_{i=1,\ldots,\kappa} \mathrm{Roll}(C_i)\,.
		\label{eq_deltaMOmathGENWALL}
	\end{equation} 
	Then the intersection of the dilated sets is empty: 
	$
		\bigcap_{i=1}^\kappa C_i(\delta) =  \varnothing. 
	$
	The condition~\eqref{eq_deltaMOmathGENWALL} with inequality instead of strict inequality is sufficient to guarantee that the intersection has zero volume. 
\end{theorem} 

The proof is omitted as it is similar to the proofs of Theorem~\ref{thm:main1} and a version of Theorem~\ref{thm:main2} for two balls and a half-plane,
 involving the necessary and sufficient condition
	\begin{equation} 
		\delta < \frac{1}{4} R\,
		\label{eq_deltaMOmathWALL}
	\end{equation} 
	 for $\delta>0$ and $R>0$,
	such that for all closed balls $B(\mathbf{r_i}, R_i)\subset \mathbb R^n$, $i=1,2$ with $R_i>0$ and centers $\mathbf{r}_i\in \mathbb R^n$ that have disjoint interiors and satisfy 
		$\min(R_1,R_2) \geq R$
		the dilated balls $B(\mathbf{r}_i, R_i+\delta)$ have empty triple intersection.

\section{Proof of Propositions~\ref{theorem_da} and~\ref{theorem_dR}} \label{sec_proofPROP}

Throughout this section we work in dimension $n=2$. The proof of Propositions~\ref{theorem_da} builds on explicit formulas for $\delta_{\max}$, governed by a case distinction. The guiding idea is the following. Fix three non-overlapping disks $B(\mathbf r_i, R_i)$, $i=1,2,3$. For very small $\delta$, all pairwise intersections $B(\mathbf r_i, R_i+\delta)\cap B(\mathbf r_j, R_j +\delta)$, $i\neq j$, vanish. As $\delta$ increases, typically at some point one of the pairwise intersections becomes non-empty, then another one, and finally all three of them. Once all three of them are non-empty, there are two possibilities: either the regions of pairwise intersection in turn intersect among themselves, or they do not. In the first case the region of triple intersection is in fact non-empty and we have found $\delta_{\max}$. In the second case we need to increase $\delta$ further before we reach $\delta_{\max}$; we shall see that in this case, $\delta_{\max}$ is the radius of one of the Apollonius circles. 

To avoid burdensome considerations of borderline cases,  we assume throughout this section that the centers $\mathbf r_i$, $i=1,2,3$ are not collinear (i.e., they do not lie on a common line). 
We note, however, that the results extend to  collinear centers as well if we allow for degenerate triangles with angles
 $\alpha_i=0$, $\alpha_j=0$ and $\alpha_k=\pi$. These cases are accounted for in Corollary~\ref{cor_deltacr}.

\subsection{Critical shell radius vs.\ Apollonius circles}

Remember the angles $\alpha_i$ and the side lengths $\ell_i$ of the triangle $\mathcal T$ with corners $\mathbf r_i$, $i=1,2,3$, see Fig.~\ref{fig:circle3_arb}. For $\{i,j,k\} = \{1,2,3\}$, set
\begin{equation}\label{eq_deltai}
	\delta_i:= \frac12 (\ell_i - R_k - R_j)\geq 0. 
\end{equation} 	
Clearly 
\begin{equation}  \label{eq_notripletI2}
	B(\mathbf r_j, R_j+\delta) \cap B(\mathbf r_k, R_k+\delta) \neq \varnothing\ \Leftrightarrow\ \delta \geq  \delta_i.
\end{equation}
We may label the corners of the triangle in such a way that the thresholds $\delta_i$ for pairwise intersections are ordered as 
\begin{equation} 
	\delta_1\leq \delta_2\leq \delta_3.
\end{equation}
If the triple intersection $\cap_{i=1}^3 B(\mathbf r_i, R_i+\delta)$ is non-empty, then necessarily all pairwise intersections are non-empty, therefore $\delta_{\max} = \delta_{\max}\bigl( B(\mathbf r_1, R_1+\delta), B(\mathbf r_2, R_2+\delta), B(\mathbf r_3, R_3+\delta)\bigr)$ satisfies
\begin{equation}
	\delta_{\max} \geq \max (\delta_1,\delta_2,\delta_3) = \delta_3. 
\end{equation} 

The main result of this subsection is the following dichotomy. 

\begin{proposition} \label{prop:dichotomy}
Let $B(\mathbf r_i, R_i)$, $i=1,2,3$ be three disks in $\mathbb R^2$ with pairwise disjoint interiors and non-collinear centers. Then, either one or the other of the following two cases occurs:
\begin{enumerate} 
	\item There exists a circle $\partial B(\mathbf p, \delta^*)$ centered in the interior of the triangle $\mathcal T$ with $|\mathbf p - \mathbf r_i| = R_i +\delta^*$ for $i=1,2,3$, in particular the circle is  externally tangent to the three circles $\partial B(\mathbf r_i, R_i)$. 
	We have 
	\begin{equation} 
		\delta_{\max} = \delta^*> \max(\delta_1,\delta_2,\delta_3)
	\end{equation} 
	and the triple intersection $\cap_{i=1}^3 B(\mathbf r_i, R_i +\delta_{\max})$.
	is the singleton $\{\mathbf p\}$. 
	\item There is no such circle, $\delta_{\max}$ is equal to $\max (\delta_1,\delta_2,\delta_3)$ 
	and the triple intersection is a  single point $\mathbf p$ on the boundary of the triangle. 
\end{enumerate} 
\end{proposition} 

The location of the triple intersection for the maximal depletion radius $\delta_{\max}$ is drawn in the top panel of Fig.~\ref{fig:circle3} for a range of selected configurations of three disks.
In case 1, the shell radius $\delta^*$ is the radius of one of the Apollonius circles \cite[Chapter VI]{johnson60}.
Note that our circle from case~1 is uniquely defined 
while the Apollonius problem may have up to eight solution circles, roughly because tangency merely requires  $|\mathbf p - \mathbf r_i| =  \pm (R_i \pm  \delta^*)$,
while the expression on the right-hand side must be positive (a minus sign of $R_i$ or $\delta^*$ 
corresponds to internal tangency which is not relevant here).
Furthermore, in case 2, there might very well be an Apollonius circle whose center satisfies  the appropriate conditions $|\mathbf p - \mathbf r_i| = R_i + \delta^*$ for all $i$, but that center is not in the interior of the triangle. 
 The relation between our depletion problem and the Apollonius circles is also illustrated in the bottom panel of Fig.~\ref{fig:circle3}.

\begin{lemma}\label{lem:aux1}
	If $\delta_{\max} = \delta_3 = \max(\delta_1,\delta_2,\delta_3)$, then the triple intersection $\cap_{i=1}^3 B(\mathbf r_i, R_i+\delta_{\max})$ at maximal shell radius $\delta = \delta_{\max}$ is equal to the pairwise intersection $\cap_{i=1}^2 B(\mathbf r_i, R_i+\delta_{\max})$ and consists of a single point $\mathbf p$ on the triangle side connecting $\mathbf r_1$ and $\mathbf r_2$. The point $\mathbf p$ satisfies 
	\begin{equation} \label{eq:lemaux1}
		|\mathbf p - \mathbf r_1|= R_1+\delta_{\max}, \quad  |\mathbf p - \mathbf r_2|= R_2+\delta_{\max},\quad |\mathbf p - \mathbf r_3|\leq R_3+\delta_{\max}
	\end{equation} 
\end{lemma} 

\begin{proof} 
	At $\delta = \delta_3$ the intersection $\cap_{i=1}^2 B(\mathbf r_i, R_i + \delta_3)$ consists of a single point $\mathbf p$ on the triangle connecting $\mathbf r_1$ and $\mathbf r_2$ with distance $R_1+ \delta_3$ and $R_2+ \delta_3$ to $\mathbf r_1$ and $\mathbf r_2$, respectively. If $\delta_3=\delta_{\max}$, then the triple intersection $\cap_{i=1}^3 B(\mathbf r_i, R_i+\delta_{\max})$ is non-empty and contained in the double intersection $\cap_{i=1}^2 B(\mathbf r_i, R_i+\delta_{\max}) = \{\mathbf p\}$, therefore the triple intersection is the singleton $\{\mathbf p\}$ and the point $\mathbf p$ must satisfy~\eqref{eq:lemaux1}.
\end{proof}

\begin{lemma} \label{lem:aux2} 
	If $\delta_{\max} >\delta_3= \max (\delta_1,\delta_2,\delta_3)$, then the triple intersection $\cap_{i=1}^3 B(\mathbf r_i, R_i+\delta_{\max})$ at maximal shell radius $\delta = \delta_{\max}$ consists of a single point $\mathbf p$ that lies in the interior of the triangle $\mathcal T$ and satisfies 
	\begin{equation} 
		|\mathbf p - \mathbf r_i| = R_i +\delta_{\max},\quad i = 1,2,3. 
	\end{equation} 
\end{lemma} 

Therefore the circle $\partial B(\mathbf p,\delta_{\max})$ is tangent to each of the circles $\partial B(\mathbf r_i, R_i)$ (i.e.,  it intersects each of those circles in exactly one point): it is one of the circles that solves the Apollonius problem. 

\begin{proof}
	For $\delta>\delta_3$,  the pairwise intersection $\mathcal{L}_\text{pc}(\delta) :=  B(\mathbf r_1, R_1+\delta) \cap  B(\mathbf r_2, R_2 + \delta)$ of two disks
	is a lens-shaped non-empty, compact, convex set. 
	The intersection of the lens with $B(\mathbf r_3, R_3+\delta)$ is equal to the triple intersection of the disks $B(\mathbf r_i, R_i +\delta)$, $i=1,2,3$. It is empty for $\delta< \delta_{\max}$ and non-empty for $\delta \geq \delta_{\max}$. Therefore, for $\delta < \delta_{\max}$, the lens $\mathcal L_\text{pc}(\delta)$ has distance strictly larger than $R_3 + \delta$ to $\mathbf r_3$; for $\delta= \delta_{\max}$, the lens has distance equal to $R_3+\delta_{\max}$ to $\mathbf r_3$. Let $\mathbf p$ be the uniquely defined point on the boundary of the lens that minimizes the distance towards $\mathbf r_3$, so that $|\mathbf p - \mathbf r_3| = R_3+ \delta_{\max}$. Then 
	\begin{equation}
		\bigcap_{i=1}^3 B(\mathbf r_i, R_i +\delta_{\max}) = \mathcal L_\text{pc}(\delta_{\max}) \cap B(\mathbf r_3, R_3 +\delta_{\max}) = \{\mathbf p\}.
	\end{equation} 
	The point $\mathbf p$ lies on the boundary of the lens, therefore $|\mathbf p - \mathbf r_i|\leq R_i +\delta_{\max}$ for $i=1,2$ with equality in at least one of the indices $1,2$. Suppose by contradiction that equality holds for exactly one index, say $|\mathbf p - \mathbf r_2| = R_2+ \delta_{\max}$ but $|\mathbf p - \mathbf r_1| < R_1 + \delta_{\max}$. This means that $\mathbf p$ is not a corner of the lens. 
	Then $\mathbf p$ must be on the intersection of $\partial B(\mathbf r_2, R_2+\delta_{\max})$ with the segment $[\mathbf r_2, \mathbf r_3]$. This implies $\delta_{\max} = \delta_1$, contradicting the assumption $\delta_{\max} > \max(\delta_1,\delta_2,\delta_3)$. Therefore $\mathbf p$ is a corner of the lens and $|\mathbf p - \mathbf r_i| = R_i+\delta_{\max}$ for all $i\in \{1,2,3\}$. 
	
	Finally, suppose by contradiction that the lens corner $\mathbf p$ is not in the interior of the triangle. If it is on the boundary, then it must be on the triangle side connecting $\mathbf r_1$ and $\mathbf r_3$, hence $\delta_{\max} = \delta_2$ in contradiction with $\delta_{\max}>\delta_3\geq \delta_2$. If $\mathbf p$ is outside $\mathcal T$, then the lens must contain a part of the segment $[\mathbf r_1, \mathbf r_3]$ in its interior. The intersection of $\partial B(\mathbf r_1, R_1+\delta_{\max})$ with the segment consists of a single point $\mathbf q$ that is closer to $\mathbf r_3$ than $\mathbf p$, contradicting the definition of $\mathbf p$. 
\end{proof} 

The following is a converse to the previous lemma. 

\begin{lemma} \label{lem:aux3}
	Suppose that there exist a point $\mathbf p$ in the interior of the triangle $\mathcal T$ and a radius $\delta^*>0$ such that $|\mathbf p - \mathbf r_i| = R_i +\delta^*$ for $i=1,2,3$. Then $\delta_{\max}$ is equal to $\delta^*$ and the triple intersection of the disks $B(\mathbf r_i,R_i+\delta_{\max})$ consists of the unique point $\mathbf p$. 
\end{lemma} 

\begin{proof}
	The point $\mathbf p$ is in the triple intersection $\cap_{i=1}^3 B(\mathbf r_i, R_i + \delta^*)$, therefore the latter is non-empty and $\delta_{\max} \leq \delta^*$. At $\delta = \delta^*$, the double intersections $\mathcal L_{ij}(\delta) = \cap_{s\in \{i,j\}} B(\mathbf r_s, R_s+\delta)$ are lens-shaped regions, centered on the triangle sides, that meet at their tips in $\mathbf p$, so the triple intersection consists of exactly the point $\mathbf p$. At $\delta < \delta^*$ the lenses are smaller and do not meet at all, hence the triple intersection is empty and $\delta_{\max}$ is actually equal to $\delta^*$. 
\end{proof} 

Proposition~\ref{prop:dichotomy} follows from Lemmas~\ref{lem:aux1}--\ref{lem:aux3}, the details are left to the reader. 

\subsection{Monotonicity. Proof of Proposition~\ref{theorem_da}}

Let us investigate how the maximal shell radius $\delta_{\max}$ changes when we vary the angle $\alpha_1$ at fixed side lengths $\ell_2$ and $\ell_3$ (without loss of generality). We distinguish three cases: 
\begin{enumerate} 
	\item $\delta_{\max} = \max(\delta_2,\delta_3)$, i.e., we are in case 2 of Prop.~\ref{prop:dichotomy} with triple intersection point on one of the triangle sides adjacent to $\mathbf r_1$. 
	\item  $\delta_{\max}= \delta^*>\max (\delta_1,\delta_2,\delta_3)$ with triple intersection point in the interior of the triangle (case 1 of Prop.~\ref{prop:dichotomy}).
	\item $\delta_{\max} = \delta_1$, i.e., triple intersection point on the triangle edge opposite $\mathbf r_1$ (back to case 2 of Prop.~\ref{prop:dichotomy}). 
\end{enumerate} 
We shall see in Steps 4--5 below that, as $\alpha_1$ increases, the cases occur in the order listed above. That is, for small $\alpha_1$ the point of triple intersection may be initially on one of the triangle sides adjacent to $\mathbf r_1$ and stay there for small $\alpha_1$. With increasing $\alpha_1$ the point moves to the interior of the triangle until it hits the side opposite $\mathbf r_1$; once there, it does not leave. 

The proof of Prop.~\ref{theorem_da} is in several steps. First we establish monotonicity within the three cases (steps 1--3), then we address the order of occurrence of the three cases (steps 4--5), and finally  we conclude using Prop.~\ref{prop:dichotomy}. 

\textbf{1.} Monotonicity within case 1 is trivial as $\delta_2$ and $\delta_3$ do not depend on $\alpha_1$. 

\textbf{2.} Monotonicity within case 3 is easily established as well. The threshold
$\delta_1 = \frac12 (\ell_1 - R_2 - R_3)$ depends on $\alpha_1$ via $\ell_1$ given by
\begin{equation}
	\ell_1^2 = \ell_2^2 + \ell_3^2 - 2 \ell_2 \ell_3 \cos \alpha_1.
\end{equation} 	
Hence $\ell_1$ and $\delta_1= \delta_1(\alpha_1)$ are increasing functions of $\alpha_1\in (0,\pi)$. 

\textbf{3.} For monotonicity within case 2 of the list above, let $\mathbf p$ be the triple intersection point in the interior of the triangle. The ray from $\mathbf r_1$ to $\mathbf p$ splits the angle $\alpha_1 \in (0,\pi)$ into two angles $\alpha_1^{(2)}>0$ and $\alpha_1^{(3)}>0$ (see Fig.~\ref{fig:circle3_arb}), 
\begin{equation} \label{eq:alpha1}
	\alpha_1 = \alpha_1^{(2)}+ \alpha_1^{(3)}.
\end{equation}
By the cosine rule in the triangle with vertices $\mathbf r_1, \mathbf r_2, \mathbf p$, the angle $\alpha_1^{(2)}$  satisfies
\begin{equation} \label{eq:alpha12}
	\cos(\alpha_1^{(2)}) = \frac{\ell_3^2 + (R_2+ \delta^*)^2 - (R_1+\delta^*)^2}{2 \ell_3 (R_1+\delta^*)} =: f_2(\delta^*)
\end{equation} 
Similarly, 
\begin{equation} \label{eq:alpha13}
	\cos(\alpha_1^{(3)}) = \frac{\ell_2^2 + (R_3+ \delta^*)^2 - (R_1+\delta^*)^2}{2 \ell_2 (R_1+\delta^*)}
	 =: f_3(\delta^*).
\end{equation}
The derivative of the right-hand side of~\eqref{eq:alpha12} with respect to $\delta^*$ (at fixed $R_1,R_2$ and $\ell_3$) can be computed explicitly, it is equal to 
\begin{equation}
	f'_2(\delta^*) = \frac{(R_1- R_2)^2 - \ell_3^2}{2 \ell_3(R_1+\delta^*)}< 0.
\end{equation} 
It is strictly negative because $\ell_3 \geq R_1+ R_2 > |R_1-R_2|$. Furthermore  $f_2(\delta_3) = 1$ and 
\begin{equation}
	\lim_{\delta^*\to \infty} f_2(\delta^*) = \frac{R_2- R_1}{\ell_3}\in (-1,1).
\end{equation}
It follows that the solution to Eq.~\eqref{eq:alpha12}, given by 
\begin{equation} \label{eq:arcos1}
	\alpha_1^{(2)}(\delta^*):= \arccos f_2(\delta^*)\in [0,\pi) \quad (\delta^*\geq \delta_3) 
\end{equation} 
is a strictly increasing function of $\delta^*$. The same holds true for
\begin{equation} \label{eq:arcos2}
	\alpha_1^{(3)}(\delta^*):= \arccos f_3(\delta^*)\in [0,\pi) \quad (\delta^*\geq \delta_2).
\end{equation} 
Therefore Eqs.~\eqref{eq:alpha1}--\eqref{eq:alpha13} determine the angle $\alpha_1$ uniquely as a strictly increasing function of $\delta^*$. Inverting the relationship, we see that $\delta^* = \delta^*(\alpha_1)$ is an increasing function of $\alpha_1$. Monotonicity within case 2 of the list above follows.

\textbf{4.} Now we turn to the order in which the cases occur. We may assume without loss of generality $\delta_3 \geq \delta_2$. Let $\mathbf q$ be the point of contact of $B(\mathbf r_1, R_1)$ and $B(\mathbf r_2, R_2)$. Then case 1 occurs if and only if $|\mathbf q - \mathbf r_3|\leq R_3 + \delta_3$. Because of
\begin{equation}
	|\mathbf q - \mathbf r_3|^2 = (R_1+ \delta_3)^2 + \ell_2^2 - 2 (R_1+\delta_3) \ell_2 \cos \alpha_1,
\end{equation} 
the distance $|\mathbf q - \mathbf r_3|$ increases with $\alpha_1$. Therefore, once $|\mathbf q - \mathbf r_3|$ crosses the threshold $R_3 +\delta_3$ it stays above the threshold.  That is, once we have left case 1 we do not come back to it.  

\textbf{5.} Similarly, once we leave case 2 while increasing $\alpha_1$, we do not come back to it. To see why, we revisit the computations from step 3. Let  $\alpha_1\in (0,\pi)$ be an angle that falls into case 2. Then, as proven earlier, $\delta_{\max}(\alpha_1)  = \delta^* (\alpha_1)$ is a solution to Eqs.~\eqref{eq:alpha1}--\eqref{eq:alpha13}. Moreover, because the center $\mathbf p$ of the circle $\partial B(\mathbf p,\delta^*)$ is inside the triangle $\mathcal T$, the inner angle $\beta_1$ at $\mathbf p$ in the triangle with vertices $\mathbf p$, $\mathbf r_2$, $\mathbf r_3$ is smaller than $\pi$. The angle $\beta_1 = \beta_1(\alpha_1)$ is an increasing function of $\alpha_1$. This is because the angles $\alpha_2^{(1)}$ and $\alpha_3^{(1)}$ are increasing functions of $\delta^*$, by arguments similar to step 3, and 
\begin{equation} 
	\beta_1 = 2\pi - \alpha_1 - \alpha_2^{(1)} - \alpha_3^{(1)}.
\end{equation} 
Conversely, let $\alpha_1$ be an angle with $\alpha_1 = \alpha_1^{(2)}(\delta^*) + \alpha_1^{(3)}(\delta^*)=:\alpha_1(\delta^*)$ for some $\delta^*>\max (\delta_2,\delta_3)$, with $\alpha_1^{(i)}(\delta^*)$ defined as in Eqs.~\eqref{eq:arcos1} and~\eqref{eq:arcos2}. The angles $\alpha_1^{(i)}(\delta^*)$ allow us to set a point $\mathbf q$ that satisfies $|\mathbf q - \mathbf r_i| = R_i + \delta^*$ for $i=1,2,3$. In general the point $\mathbf q$ might be outside the triangle or on its boundary, but if $\beta_1<\pi$, then the point $\mathbf q$ is in the interior of the triangle and we are in case 2. 

Fix an angle $\alpha_1$ that falls within case 3 and pick $\alpha'_1>\alpha_1$. We want to show that $\alpha'_1$ falls into case 3 as well. By step 4 we know that $\alpha'_1$ cannot fall into case 1, it remains to exclude case 2.  The map $\delta^*\mapsto \alpha_1(\delta^*)$ is a strictly increasing bijection from $(\max(\delta_2,\delta_3),\infty)$ onto some subinterval of $(0,\infty)$. Therefore, if the equation $\alpha_1(\delta^*)= \alpha_1$ has no solution, the same holds true for all $\alpha'_1>\alpha_1$ and $\alpha'_1$ cannot fall into case 2. If the equation $\alpha_1(\delta^*) = \alpha_1$ has a solution, then by the considerations above we must have $\beta_1(\alpha_1) \geq \pi$ hence $\beta_1(\alpha'_1)\geq \beta_1(\alpha_1)\geq \pi$ and case $2$ is excluded as well. Thus we have proven that if $\alpha_1$ is in case 3, then all larger angles $\alpha'_1\geq \alpha_1$ are in case 3 as well. 

\textbf{6.} The monotonicity of the map $\alpha_1\mapsto \delta_{\max}(\alpha_1,\ell_2,\ell_3)$ and, generally, $\alpha_i\mapsto \delta_{\max}(\alpha_i,\ell_j,\ell_k)$ follows from Steps 1--5 and the continuity of the map. \hfill $\qed$

\subsection{Descartes' circle theorem. Proof of Proposition~\ref{theorem_dR}}

Consider three circles with centers $\mathbf r_i$ and radii $R_i>0$ that are in contact, i.e., $\ell_1 = R_2+ R_3$ and similarly for the other triplets. In this situation the pairwise thresholds $\delta_i$ vanish and in view of $\delta_{\max}>0$, we are automatically in case 1 of Proposition~\ref{prop:dichotomy}. 

Thus $\delta_{\max}  =\delta^*$ is the radius of a circle $\partial B(\mathbf p,\delta^*)$ centered in the interior of the triangle and tangent to the three adjacent circles $\partial B(\mathbf r_i, R_i)$. 

This is precisely the special case of the Apollonius problem addressed by Descartes' circle theorem.   There are exactly two circles adjacent to the three circles $\partial B(\mathbf r_i, R_i)$.
 There always exists one solution circle which is inscribed inside the void delimited by the three given circles 
while the other solution circle either contains all three of them (bottom panel of Fig.~\ref{fig:circle3}a) or forms a void together with two of the given circles containing the third one (bottom panel of Fig.~\ref{fig:circle3}b).
Clearly the solution that interests us is the inner circle, called inner Soddy circle. Its radius is known (see for example \cite{coxeter68,oldknow96}), it is equal to the expression for $\delta_{\mathrm{cr}}(R_1,R_2,R_3)$ given in Proposition~\ref{theorem_dR}. This completes the proof of the proposition. \hfill $\qed$

\section{Physical interpretation of geometric criteria \label{sec_conclusions}}

To summarize, we have established rigorous geometric criteria for the absence of many-body interactions higher than pairwise
in nonadditive AO-type mixtures involving hard colloids and ideal depletants in any spatial dimension.
Once the stated criteria are obeyed, our results imply that the depletants can be integrated out and there
exists an exact mapping onto a system with effective pairwise interactions only.
 Mathematically, the absence of triplet and $N$-body interactions with $N\geq3$ in a fluid of hard bodies for a given depletion radius
 is equivalent to the triple intersection of the corresponding dilated bodies (and thus also the intersections between more than three dilated bodies) having zero volume for any possible configuration of the hard bodies.  
 We elucidated relevant conditions for this problem in Theorems~\ref{thm:main1} and~\ref{thm:main2}
 and proved them with elementary geometry.

For the exactness of the depletion interaction between {\it identical colloidal hard spheres}, 
the sufficient and necessary geometric criterion, Eq.~\eqref{eq_deltaMO}, was
already stated in the literature but is mathematically proven here as a special case of Eq.~\eqref{eq_deltaMOmath} with all radii  being equal, i.e., for $R=R_1=R_2=R_3$.

The generalized condition in our Theorem~\ref{thm:main2} holds for {\it polydisperse mixtures of hard spheres} in any spatial dimension, where the radius $R=\min_{i}R_i$ of the smallest species sets the threshold according to Eq.~\eqref{eq_deltaMOmath}.

Moreover, we provide a condition, Eq.~\eqref{eq_deltaMOmathGEN}, which is sufficient for the exactness of the pairwise depletion interaction between {\it convex hard particles}, also including {\it polydisperse mixtures} thereof.
The generalized upper bound for the depletion radius provided in Theorem~\ref{thm:main1} is proportional to the minimal rolling radius among all bodies,
 which corresponds to the curvature radius at the point on the surface of any body which has the largest curvature.
This criterion may, however, turn out to be extremely susceptible to irregularities of the surface geometry (a detailed discussion can be found below in Sec.~\ref{sec_outlook}).
 We therefore note that a systematic improvement of the threshold in Eq.~\eqref{eq_deltaMOmathGEN} can be achieved, along the lines of Remark~\ref{rem:improvedcriterion},
 by applying  Theorem~\ref{thm:main1} to appropriately chosen auxiliary bodies (with larger rolling radii) that are contained in the original bodies.

By slightly reformulating our main theorems, we also address the
problem of effective depletion interactions between {\it arbitrarily-shaped convex particles and a planar hard wall} in Theorem~\ref{thm:main1WALL}.
In this case, a one-body depletion potential can be found from integrating out the depletants.
Such a potential is exact if there are no effective pair and $N$-body interactions with $N\geq2$, 
which is again equivalent to the triple intersection of the corresponding dilated bodies (including that of the wall) having zero volume.
The sufficient and necessary geometric criterion, Eq.~\eqref{eq_deltaPL}, for identical spheres
has been generalized to Eq.~\eqref{eq_deltaMOmathWALL}, which again shows that the radius of the smallest species determines the critical depletion radius
 for polydisperse mixtures of hard spheres.
Relatedly, Eq.~\eqref{eq_deltaMOmathGENWALL} provides a sufficient criterion for general polydisperse mixtures based on the minimal occurring curvature radius.

 Although not explicitly stated in our theorems, our mathematical analysis directly applies to
 {\it polydisperse spherical depletants}, for which the depletion radius $\delta^{(\nu)}$ may depend on the species $\nu$.
 In this case, our criteria should be interpreted to provide thresholds for the maximal depletion radius $\delta=\max_{\nu}\delta^{(\nu)}$ associated with the ``largest'' depletant species.

{\it Non-spherical depletants} such as needles or rods, which have been widely considered in the literature \cite{Mao1,Mao2,Schweizer,Lang,mueller,Lettinga,opdam2022excluded},
give rise to a nonspherical shape of the depletion shell which depends on their orientation.
The thresholds for $\delta$ can thus directly be applied as sufficient criteria on the maximal radius of such a generalized depletion shell.
However, considering the resulting generalized geometric overlap problem in its full complexity, a better upper bound could be found for some particular situations in future work.

\section{Conclusions \label{sec_outlook}}

Let us conclude with a couple of remarks regarding the broader physical implications of our criteria and possible generalizations.

First of all, if the convex colloidal particles have a {\it cusp\/}
 in their shape, like colloidal polyhedra \cite{Vroege,Marechal2},
our sufficient criteria become trivial since the rolling radius vanishes (diverging local surface curvature).
The resulting condition $\delta<0$ can never be fulfilled by a positive depletion radius $\delta$,
such that the absence of higher-body depletion interactions can never be guaranteed.
Indeed, many regular particle shapes, e.g., tetrahedral or cubic, 
can arrange in configurations with the cusps of three or more particles at direct contact.
 In these special cases, it is apparent that any choice of depletion radius $\delta>0$ results in overlapping depletion shells, 
so that our criteria even become necessary.
In contrast, if configurations with three touching cusps are forbidden, e.g., for particles with sufficiently large opening angles, the formal condition $\delta<0$ apparently is no longer sharp.
 This extreme example illustrates that the maximal depletion radius can strongly depend on the particle shape.
In practice, this is also relevant for colloidal particles, which are always slightly rounded \cite{Vroege}.
Regarding the convenient possibility to approximate the shape of smooth bodies by appropriate polygons, particles with cusps do indeed remain relevant for theoretical treatments.
 In both cases of singular cusps or smooth irregularities of the surface, an improved sufficient bound can be implicitly specified according to Remark~\ref{rem:improvedcriterion}.
The open question of whether this criterion is also necessary remains an interesting problem for future work.

Second, {\it non-convex\/} hard bodies play an important role in modeling lock-and-key colloids \cite{Pine1,Pine2}
where the depletion attraction has been used as a bonding mechanism \cite{Pine1,Jack}.
Other relevant nonconvex shapes are clusters of spheres firmly attached to each other by surface chemistry
\cite{Kraft,Kegel}, colloidal bowls \cite{Marechal,Jack,Pine3} or even more general objects, see Ref.~\cite{Wittkowski} for some examples.
 As stated in Remark \ref{rem:convexity}, our sufficient criteria do also apply for nonconvex shapes.
However, the provided thresholds could, in principle, become arbitrarily poor, imagining, for example particles possessing a cusp which is inaccessible to other particles.
 Therefore, we leave an improvement of our bounds for such particles to future studies.

Third, another important problem which we did not explicitly consider here is the depletion at {\it curved hard walls}.
In the case of externally placed obstacles \cite{interface_AO_3}, which provide a model for a porous medium,
a criterion for the absence of effective external pair interactions can be found analogously to case of a planar hard wall
by considering Eq.~\eqref{eq:mindelta-crWALL} without sending the wall radius $R_\text{w}$ to infinity.
We further expect that the structure of our proof can also be applied to rigorously determine
an explicit formula for the critical depletion radius for a fluid in external hard-body confinement \cite{interface_AO_4}.
In this case, we expect that such a threshold is related to a solution of the Apollonius problem involving internal tangency.
 
Fourth, one may consider more realistic situations of {\it softened colloid--polymer interactions} \cite{Louis0,Fuchs}
or non-vanishing polymer--polymer interactions \cite{Louis0,Chervanyov}.
Only for colloid--polymer interactions which exhibit a sharp exclusion
zone of radius $\delta$ plus an interaction of strict finite range $\epsilon$
and for ideal polymers, the pure geometric concept applies
such that our criteria can be used for the overall range $\delta+\epsilon$ of the colloid--polymer interaction. Long-ranged colloid--polymer interactions
will induce many-body interactions even for ideal polymers. Likewise, non-vanishing
polymer--polymer interactions will in general contribute to many-body interactions of any order
due to their finite correlation length.

Fifth, recent studies have included {\it activity\/}  in the depletants leading to
active depletion interaction between colloids \cite{Reichhardt,Ni,Harder,Smallenburg}.
For non-additive hard-particle mixtures and ideal active depletants, the generalization of the AO
pair interaction was computed recently \cite{Smallenburg} using the concept of swim pressure.
It is important to mention here that our proof for the absence of many-body interactions
does also apply to active ideal depletants, as the basic geometric conditions are identical.
The only difference is that for active depletants
the osmotic swim pressure depends on the shape and local curvature
 of the excluded zone but this only affects the strength of the pair interactions but not
the geometric conditions for triple intersection.
On a further note, the effective interactions between active colloids display intriguing similarities with effective depletion interactions \cite{turci2021phase}
but are generally of intrinsic many-body nature \cite{kaiser2012capture,farage2015effective,rein2016applicability,wittmann2017effective,turci2021phase}.

Finally, beyond verifying the absence of effective triplet interactions in the AO model, the {\it magic number\/} $2/\sqrt{3}-1\approx0.1547$ 
is in fact a more general indicator of crossover behavior in hard-particle systems.
For example, the close-packing density in a binary mixture of hard spheres changes its behavior as a function of the size ratio around this value \cite{mcgearyHSpacking1961}.
 When the size ratio is below this threshold, there even exist dynamical escape routes for the smaller particles preventing their vitrification~\cite{binaryglassmagic}.
Our presented ideas could thus be also helpful to deepen the mathematical understanding of such related physical problems.

\section*{Acknowledgments}

The authors would like to thank R.~D.~Mills-Williams for helpful discussions at the initial stage of the project, D.~Hug for stimulating comments, J.~Horbach for pointing out relevant references and 
G.~H.~P.~Nguyen for aid in creating the graphics. Special thanks goes to M.~A.~Klatt for suggesting Remark~\ref{rem:improvedcriterion} after carefully reading the manuscript.
Finally, we gratefully acknowledge support by the Deutsche Forschungsgemeinschaft (DFG) through the SPP 2265, under grant numbers WI 5527/1-1 (R.W.),  JA 2511/2-1 (S.J.) and LO 418/25-1 (H.L.).

\end{document}